\documentclass[conference,letterpaper]{IEEEtran}

\usepackage[utf8]{inputenc} 
\usepackage[T1]{fontenc}
\usepackage{url}
\usepackage{ifthen}
\usepackage{cite}
\usepackage{graphicx}
\usepackage{amssymb}
\usepackage[cmex10]{amsmath} 
\usepackage{xcolor}
\usepackage {hyperref}
\usepackage{amsthm}

\newtheorem{definition}{Definition}[]

\newtheorem{thm}{Theorem}[]
\newtheorem{lem}[thm]{Lemma}
\newtheorem{prop}[thm]{Proposition}
\newtheorem{cor}[thm]{Corollary}

\newtheorem{rem}{Remark}

\begin{document}
\title{Bounds on the Statistical Leakage-Resilience of Shamir's Secret Sharing} 

% %%% Single author, or several authors with same affiliation:
% \author{%
%  \IEEEauthorblockN{Andrew R.~Barron}
%  \IEEEauthorblockA{Department of Statistics and Data Science\\
%                    Yale University\\
%                    New Haven, CT, USA\\
%                    Email: andrew.barron@yale.edu}
% }

%%% Several authors with up to three affiliations:
\author{%
  \IEEEauthorblockN{Utkarsh Gupta and Hessam Mahdavifar}
  \IEEEauthorblockA{Department of Electrical and Computer Engineering\\
                    Northeastern University\\
                    Boston, MA, USA\\
                    Email: \{gupta.utk, h.mahdavifar\}@northeastern.edu}
  % \and
  % \IEEEauthorblockN{H}
  % \IEEEauthorblockA{Bell Telephone Laboratories, Inc.\\ 
  %                   Murray Hill, NJ, USA\\
  %                   Email: \{csh, dsl\}@bell-labs.com}
}

%%% Many authors with many affiliations:
% \author{%
%   \IEEEauthorblockN{Andrew R.~Barron\IEEEauthorrefmark{1},
%                     Claude E.~Shannon\IEEEauthorrefmark{2},
%                     David Slepian\IEEEauthorrefmark{2},
%                     and Jacob Ziv\IEEEauthorrefmark{2}\IEEEauthorrefmark{3}}
%   \IEEEauthorblockA{\IEEEauthorrefmark{1}%
%                    Department of Statistics and Data Science, Yale University, New Haven, CT, USA,
%                     andrew.barron@yale.edu}
%   \IEEEauthorblockA{\IEEEauthorrefmark{2}%
%                     Bell Telephone Laboratories, Inc.,
%                     Murray Hill, NJ, USA,
%                     \{csh,dsl,jz\}@bell-labs.com}
%   \IEEEauthorblockA{\IEEEauthorrefmark{3}%
%                     Department of Electrical Engineering, Technion---Institute of Technology, Haifa, Israel,
%                     jz@ee.technion.ac.il}
% }

\maketitle

%%%%%%
%% Abstract: 
%% If your paper is eligible for the student paper award, please add
%% the comment "THIS PAPER IS ELIGIBLE FOR THE STUDENT PAPER
%% AWARD." as a first line in the abstract. 
%% For the final version of the accepted paper, please do not forget
%% to remove this comment!
%%

%abstract will be the summary of the introduction
\begin{abstract}
Secret sharing is an instrumental tool for sharing secret keys in distributed systems. In a classical threshold setting, this involves a dealer who has a secret/key, a set of parties/users to which shares of the secret are sent, and a threshold on the number of users whose presence is needed in order to recover the secret. In secret sharing, secure links with no leakage are often assumed between the involved parties. However, when the users are nodes in a communication network and all the links are physical links, e.g., wireless, such assumptions are not valid anymore. In order to study this critical problem, we propose a \textit{statistical leakage} model of secret sharing, where some noisy versions of all the secret shares might be independently leaked to an adversary. We then study the resilience of the seminal Shamir's secret sharing scheme with \textit{statistical leakage}, and bound certain measures of security (i.e., semantic security, mutual information security), given other parameters of the system including the amount of leakage from each secret share. We show that for an \textit{extreme scenario} of Shamir's scheme, in particular when the underlying field characteristic is $2$, the security of each bit of the secret against leakage improves exponentially with the number of users. To the best of our knowledge, this is the first attempt towards understanding secret sharing under general statistical noisy leakage.

\end{abstract}

\section{Introduction}
Secret sharing, introduced by Shamir \cite{shamir1979share}, and Blakey \cite{blakley1979safeguarding}, is a fundamental cryptographic primitive central to security in many distributed systems. Secret sharing protects against collusion by allowing a secret to be shared among parties/users in such a way that only some selected subsets of users can recover the secret by aggregating their shares together. Such schemes have found several applications, such as in multi-party computation\cite{goyal2022multi, patel2016secure, chaum1988multiparty, wigderson1988completeness}, zero-knowledge proofs \cite{vaikuntanathan2015secret, bendlin2010threshold}, threshold cryptographic systems \cite{shoup2000practical, desmedt1992threshold, desmedt1991shared}, access control \cite{naor1996access}, generalized oblivious transfer \cite{tassa2011generalized, shankar2008alternative}, and others. Secret sharing is widely used, yet it is still not fully understood how communication constraints such as leakage, noise, and scalability impact the security and reliability of such schemes.

%don't really need to add IoT
The massive deployment of communication networks makes information security unprecedentedly important \cite{zou2016survey}. Anyone within the geographical range of a communication network can potentially gain unauthorized access to the data transmitted over the physical links. Protocols involving secret sharing often assume that the dealer has a reliable and secure communication channel to all the parties (users in a wireless network), i.e. links with no leakages. However, real-world implementations are susceptible to side-channel attacks which may lead to the adversary gaining some information about the secret shares of the non-colluding (honest) users. 
%Therefore, an assumption is made that the adversary controlling an unauthorized set of parties gets no information about the shares of the honest parties. Interconnected wireless devices must be able to securely exchange confidential messages while being at a higher risk of malicious attacks. 

Motivated by the emergence of such attacks, protocols which attempt to provide provable security guarantees against information leakage, have attracted a lot of attention in the cryptographic community (survey \cite{kalai2019survey}). In the context of secret sharing, wireless leakage attacks allow an adversary to obtain some \textit{bounded} leakage from the secret shares of honest parties. Such leakage may help the adversary reconstruct the secret. In the past few years, substantial research has examined the feasibility and efficiency of leakage-resilient secret sharing against diverse models of potential leakages \cite{hoffmann2023stronger, klein2023new, maji2022leakage, maji2022tight, adams2021lower, benhamouda2021local, maji2021constructing, maji2021leakage, nielsen2020lower, benhamouda2018on}. 

Initiated by the fundamental work of Wyner on wiretap channels \cite{wyner1975wire}, developing information-theoretic methods for secure communications has been an active area of research \cite{angueira2022survey, liu2016physical, mukherjee2014principles}. The wiretap channel is a model of information leakage used to study physical layer security in wireless communications. In this work (Section~\ref{sec: Leakage Model}), we propose a statistical leakage model for secret sharing schemes where the adversary leaks information from all secret shares through independent wiretap channels. Subsequently, in Section~\ref{sec: Security_Analysis}, we study the resilience of Shamir's secret sharing scheme with statistical leakage. For the Shamir's scheme with a general threshold, ShamirSS$(N,t)$, we show that the mutual information leaked from the secret is less than the sum of mutual information leaked by the shares individually. Then, for an extreme scenario of the Shamir's scheme, ShamirSS$(N,N)$, in particular when the underlying field characteristic is $2$; we show that the mutual information leaked from each bit of the secret reduces exponentially with the number of users. 

The rest of this paper is organized as follows. We discuss the preliminaries to the Shamir's secret sharing protocol, and recall the relationship among information-security metrics for the wiretap channel in Section~\ref{sec:Preliminaries}. Then, we describe the proposed statistical leakage model and discuss some related leakage models in Section~\ref{sec: Leakage Model}. In Section~\ref{sec: Security_Analysis} we present results on the leakage resilience of Shamir's secret sharing under the proposed leakage model. Finally, the paper is concluded in Section ~\ref{sec:conc}.

%The title of your section 2
\section{Preliminaries}
\label{sec:Preliminaries}
%The title of your first subsection in section 2
In this section, we first review the Shamir's secret sharing scheme \cite{shamir1979share, blakley1979safeguarding}. The scheme provably leaks no information from the secret even when an adversary has access to secret shares of some (less than a threshhold $t$) users. Subsequently, we discuss the wiretap channel \cite{wyner1975wire}; and the relationship among mutual information security (MIS), semantic security (SS), and distinguishing security (DS) metrics for this channel \cite{bellare2012semantic}. The new leakage model proposed in Section~\ref{sec: Leakage Model} is inspired by the wiretap channel, and the relationship between the metrics is used to provide security guarantees in Section~\ref{sec: Security_Analysis}.  

% In section~\ref{sec: Leakage Model}, we propose a new leakage model in which every bit of every secret share leaks information. This leakage model is inspired by the wiretap channel \cite{wyner1975wire} where the adversary eavesdrops on the communication channel output through another degraded channel. In \cite{bellare2012semantic}, the authors define the Mutual Information Security (mis) metric for the wiretap channel and establish relationship of the widely used semantic security (ss) and distinguishing security (DS) metrics. In this section we will revisit the wiretap channel and this relationship between the security metrics. 

\subsection{Shamir's Secret Sharing Scheme}
\label{subsec:ShamirSS}
Given a secret $s \in \mathbb{F}_q$, Shamir's secret sharing scheme aims to distribute secret shares $s_1,s_2, \ldots, s_N \in \mathbb{F}_q$ among $N$ users in such a way that
\begin{itemize}
    \item the secret $s$ can be reconstructed given any $t$ (threshold) or more secret shares;
    \item the knowledge of less than $t$ secret shares does not reveal any information about the secret.
\end{itemize}

Such a scheme with $N$ users and threshold $t$, denoted as ShamirSS$(N,t)$, achieves this by constructing a $(t-1)$-degree random polynomial $P(x) \in \mathbb{F}_q[x]$ written as
\begin{equation}
    P(x) = s+p_1x + \cdots + p_{t-1}x^{t-1}
\end{equation}
where the coefficients of the polynomial $p_1,\cdots , p_{t-1}$ are i.i.d. variables selected uniformly at random from $\mathbb{F}_q$. Once the polynomial is constructed, the secret share $s_i$ of user $i \in [N]$ is the evaluation of the polynomial $P(x)$ at distinct predetermined values $\gamma_i \in \mathbb{F}_q$, i.e., $s_i = P(\gamma_i)$. The coefficients of $P(x)$ are determined uniquely by its evaluation on $t$ distinct points, i.e., given at least $t$ of the evaluation points $\{\gamma_1, \cdots ,\gamma_N \}$; and their corresponding $t$ secret shares. Then $s = P(0)$ is reconstructed. Using Lagrange's polynomial interpolation formula, given $t$ evaluation points $\gamma_{e_1}, \ldots \gamma_{e_t}$;
\begin{equation}\label{eqn: ShamirSS(N,t) coefficients}
    P(0) = (-1)^{t-1}\sum_{i=1}^{t} \frac{\prod_{j =1, j \neq i}^{t} \gamma_{e_j}}{\prod_{j =1, j \neq i}^{t} (\gamma_{e_i} - \gamma_{e_j})} P(\gamma_{e_i}).
\end{equation}

%edit this remark
\begin{rem} [\cite{guruswami2016repairing}]
    This is equivalent to saying that the secret can be represented as a unique linear combination of any $t$ secret shares.
\end{rem}

\subsection{The Wiretap Channel and Security Metrics}\label{subsec: wiretap_channel}

The wiretap channel, introduced by Wyner \cite{wyner1975wire}, is a pair of communication channels $(\text{ChR, ChA})$ with the same input and output sets where $\text{ChR}$ is a channel from the sender to a receiver, and $\text{ChA}$ is the wiretap that goes from the sender to an adversary. The setting aims to analyze the security of the communicated data based solely on the assumption that the channel from sender to adversary is “noisier” than the channel from sender to receiver; the adversary considered is computationally unbounded. In their seminal work \cite{bellare2012semantic}, Bellare et. al. show that security of the communicated data does not depend on the receiver channel and establish a relationship between the classical notions of semantic security (SS), distinguishing security (DS), and mutual information security (MIS). Let $M$ denote the message being sent, and let $\text{ChA}: \{0,1\}^m \to \{0,1\}^l$ be the adversary channel.

% \begin{definition}
% [The Wiretap Model]
%     The sender applies to her message $M$ a randomized encryption algorithm $\mathcal{E}: \{0,1\}^m \to \{0,1\}^c$ to get the $\textit{sender ciphertext}$ $X$. This is transmitted to the receiver over the receiver channel $\text{ChR}$ obtaining the $\textit{receiver ciphertext}$ $Y$, which he decrypts via some algorithm $\mathcal{D}$ to recover the message $M$. The eavesdropping adversary's wiretap is modeled as another channel $\text{ChA}$, and she accordingly gets an $\textit{adversary ciphertext}$ $Z$, from which she tries to glean whatever she can about the message $M$.
% \end{definition}
\begin{definition}[Mutual information security, \cite{bellare2012semantic}]\label{def: mutual_information_security}
     The amount of information revealed about the message $M$ can be measured using the MIS-security metric $\eta_{\text{MIS}}$ defined as
     \begin{equation}
         \eta_{\text{MIS}} = \max_{P_M} I \left(M; \text{ChA}(M) \right),
     \end{equation}
     where $P_M$ is the probability mass function of the message $M$.
\end{definition}

% \begin{definition}[Semantic security, \cite{bellare2012semantic}]
%     The gain in probability for the adversary to guess a (possibly randomized) function value $f(M)$ revealing partial information about the data compared to a priori  probability of success is measured using the ss-security metric, denoted by $\eta_{\text{ss}}$, and defined as
%     \begin{equation}
%         \eta_{\text{ss}} = \max_{f,P_M} \left( \max_\mathcal{A} \text{Pr}[\mathcal{A}(\text{ChA}(\mathcal{E}(M))) = f(M)] - \max_\mathcal{S} \text{Pr}[\mathcal{S}(m)= f(M)] \right).
%     \end{equation}
%     The gain is maximized over all adversaries $\mathcal{A}$ to achieve strategy independence. This is maximized over all functions $f$, and all possible distributions of the message $M$, that is  $P_M$ is the probability mass function of message $M$. 
% \end{definition}

% \begin{definition}[Distinguishing security, \cite{bellare2012semantic}]
%     \begin{equation}
%         \eta_{\text{DS}} = \max_{P_{M}} \textbf{SD}\left( \text{ChA}(\mathcal{E}(M_0)), \text{ChA}(\mathcal{E}(M_1)) \right).
%     \end{equation}
%     The gain is maximized over all adversaries $\mathcal{A}$ to achieve strategy independence. This is maximized over all functions $f$, and all possible distributions of the message $M$, that is  $P_M$ is the probability mass function of message $M$. 
% \end{definition}

\begin{thm}[\cite{bellare2012semantic}, Theorem 1,5] \label{thm: inequality_equivalence}  
For the wiretap channel, semantic security $\eta_{\text{SS}}$, distinguishing security $\eta_{\text{DS}}$, and mutual-information security $\eta_{\text{MIS}}$ satisfy the following inequality
\begin{equation}
\eta_{\text{SS}} \le \eta_{\text{DS}} \le  \sqrt{2\eta_{\text{MIS}}}.   
\end{equation}
\end{thm}

We refer the reader to \cite{bellare2012semantic} for the exact formulation of semantic security and distinguishing security, which are all statistical measures of security.

\section{Statistical Leakage Model}\label{sec: Leakage Model}
In this section, we formally introduce the problem of statistical leakage in secret sharing schemes. Leakage resilient secret sharing aims to guarantee the security of the secret, even if the adversary obtains partial information from the secret shares of all honest parties. In the proposed model, all honest parties leak information to the adversary independently through wiretap channels (Section~\ref{subsec: wiretap_channel}), while the adversary has the complete knowledge of the secret shares of the colluding users. Later in this section, we present some other leakage models which have been studied in cyrptography literature.  

\subsection{System Model}
%Exploring the leakage resilience of secret sharing schemes in communication systems, the wiretap setting naturally arises as a model of information leakage, facilitating the examination of the information-theoretic security of such schemes.
We study the leakage resilience of secret sharing schemes for communication systems. The wiretap setting (Section~\ref{subsec: wiretap_channel}) emerges naturally as a model of information leakage for analyzing the information-theoretic security of such schemes. Consider a secret sharing scheme ith $N$ users where the dealer shares the secret $s$, which is sampled from a random variable $S$. The leakage we consider has the following properties
\begin{itemize}
    \item some information can be leaked from every user;
    \item the information leaked to the adversary from each user is independent of the leakage from other users.
\end{itemize}
\vspace{-3mm}
\begin{figure}[htbp]
\centering
{\includegraphics[width=0.48\textwidth]{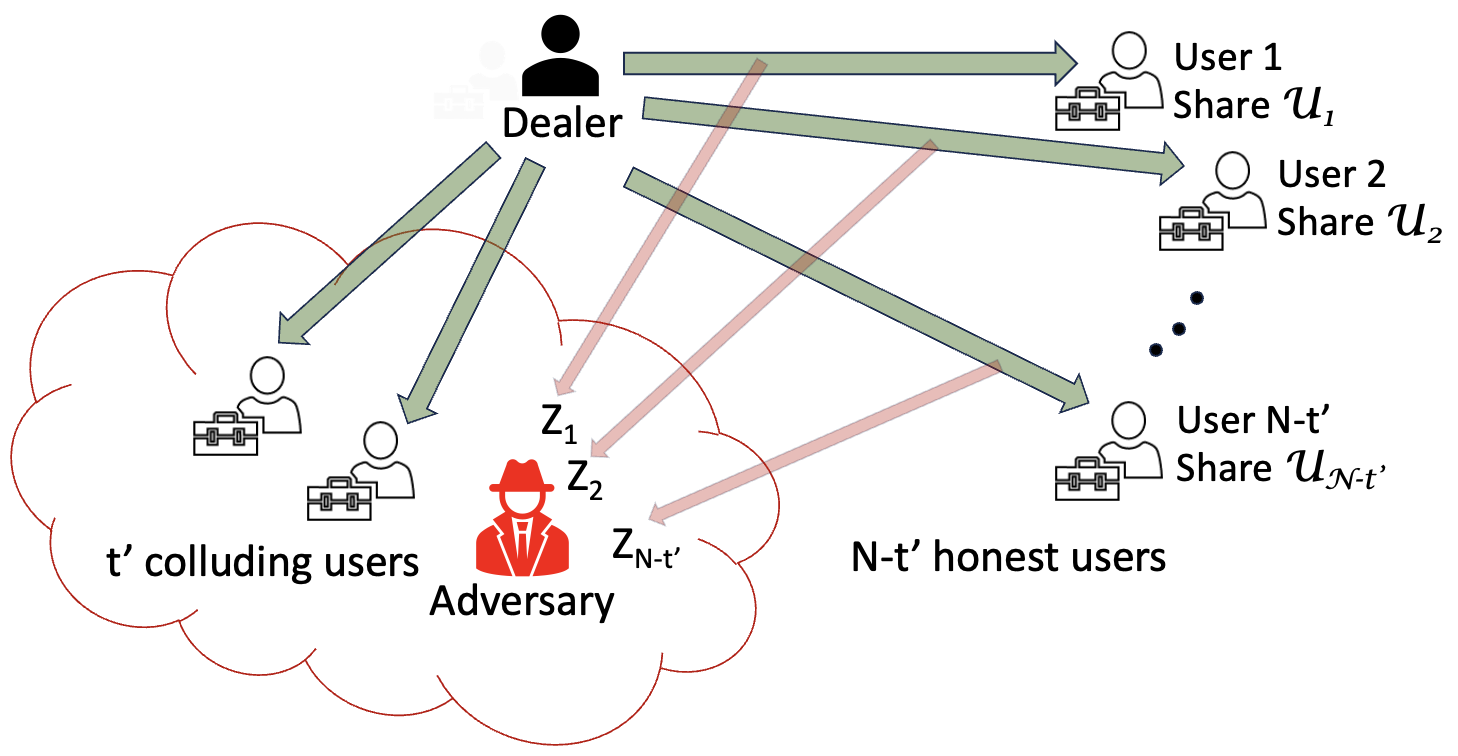}} \vspace{-3mm}
  \caption{Secret sharing with statistical leakage}\vspace{-2mm}
  \label{fig:approxifer}
\end{figure}

%The adversary leaks information from the secret shares of all honest parties through independent wiretap channels. 
The links between the dealer and users are assumed to be perfect channels while the adversary channels have information capacity constraints. Framing this model mathematically, let the random vector corresponding to the $N$ secret shares be $\mathbb{\mathcal{U}} =(\mathcal{U}_1,\mathcal{U}_2, \ldots, \mathcal{U}_N)$, and let their leaked versions be the random vector $\mathbf{Z} = (Z_1,Z_2,\ldots,Z_N)$. Let $l$ denote the bit-length of the secret and secret shares, assuming they belong to a field of characteristic $2$ and represented as vectors over the binary field. Then the proposed model is characterized by the following two assumptions:
\begin{itemize}
\item \textbf{Leakage.} $I(\mathcal{U}_j;\mathbf{Z})/l = 1$ for the colluding parties, and $I(\mathcal{U}_j;\mathbf{Z})/l \le \epsilon_j$ for the honest parties.
\item \textbf{Independence assumption.} The random variables $Z_i/\mathcal{U}_i$ i.e. the leakage channels are mutually independent.
\end{itemize}
%and provide security guarantees for the Shamir Secret Sharing Scheme with local statistical leakage of (all bits of) the secret shares. 

%In this paper, we study the security guarantees for the Shamir's secret sharing scheme ShamirSS$(N,N)$ (section~\ref{subsec:ShamirSS}) under the local statistical leakage model. In section~\ref{sec: Security_Analysis}, we give bounDS for the mutual information security (MIS), and consequently the semantic security (ss) and distinguishing security (DS) metrics for the information leakage from each bit of the secret $s$ in ShamirSS$(N,N)$. 

\subsection{Related Prior Work}
Leakage resilient secret sharing has garnered a lot of attention in the cryptographic community in recent years \cite{hoffmann2023stronger, klein2023new, maji2022leakage, maji2022tight, adams2021lower, benhamouda2021local, maji2021constructing, maji2021leakage, nielsen2020lower, benhamouda2018on}. Guruswami and Wootters' reconstruction algorithm \cite{guruswami2017repairing, guruswami2016repairing} for Reed Solomon codes showed that even a single bit leaked from each secret share compromises the security of Shamir's scheme over fields of characteristic $2$. This led Benhamouda et al. \cite{benhamouda2018on, benhamouda2021local} to investigate the leakage resilience of linear secret sharing schemes for other finite fields. They prove that Shamir’s secret sharing scheme is leakage-resilient against one bit leakages, when the underlying field is of a large prime order, and the reconstruction thresholds is at least $0.92$ times the number of parties. This threshold has been progressively refined to $0.69$ in a series of works by different authors \cite{klein2023new, maji2022improved, maji2021constructing, benhamouda2021local}.
 
Various works have tried to assess the resilience of secret sharing against diverse models of information leakage. The works of Benhamouda et al. \cite{benhamouda2018on}, and Nielsen and Simkin \cite{nielsen2020lower}, assume each secret share leaks information independently through arbitrary leakage functions with bounded output-length. Maji et. al. in \cite{maji2021leakage}, and Adams et. al. in \cite{adams2021lower}, consider probing attacks which leak physical-bits from the memory hardware storing the secret shares. In a separate work \cite{maji2022leakage}, Maji et al. consider joint leakage, where the adversary can leak any bounded output-length joint function of the shares. 

A majority of these contributions consider an adversary who is able to obtain the precise outputs of the leakage functions computed on secret shares. However, in practice, side-channel attacks are inherently noisy and there are practical techniques that can amplify this noise. Adams et al. \cite{adams2021lower} consider noisy physical-bit leakage, where each physical-bit leakage is replaced by noise with some fixed probability. In \cite{hoffmann2023stronger}, Hoffmann et. al. consider a leakage model where some random subset of the leakages is replaced by uniformly random noise. In this work, we aim at providing a general framework for secret sharing under statistical leakage from a communication/information-theoretic perspective.  %This work adds to the growing list of works in understanding resilience of secret sharing schemes with noisy leakage functions.  

% Maji et. al. \cite{maji2021leakage} considers much They prove that Shamir’s secret sharing scheme with random evaluation points is physical-bit leakage resilient if the order of the field is sufficiently large. . They prove a lower bound for the reconstruction threshold of log()/ log log() for Shamir’s secret sharing scheme, when the size of the field is 2 and the evaluation points can be chosen adversarially. In \cite{maji2022tight} Maji et al. improve their lower bound to log().

%Generally, this would allow the leakage functions to just reconstruct the secret, which is an attack that cannot be prevented. For this reason, the authors artificially restrict their leakage functions to not depend on some of the random choices made by the secret sharing scheme. For the case of Shamir secret sharing with random evaluation points, the authors show that one obtains some leakage-resilience properties, if the leakage functions are not allowed to depend on the evaluation points.

\section{Bounds on Securi ty}\label{sec: Security_Analysis}
In this section, we study the resilience of Shamir's secret sharing under the proposed statistical leakage model. For the Shamir's scheme with a general threshold, ShamirSS$(N,t)$, we show that the mutual information leaked from the secret is less than the sum of mutual information leaked by the shares individually. Then, for an extreme scenario of the Shamir's scheme, ShamirSS$(N,N)$, we show that the mutual information leaked from each bit of the secret reduces exponentially with the number of users. Let the secret being shared be denoted by the random variable $S$. For the ShamirSS$(N,N)$ scheme, we give bounds for two different cases where  
\begin{itemize}
    \item \textbf{Case 1.} the secret $S$ is uniformly distributed but the leakage channel is arbitrary (Proposition~\ref{prop: security_MGL}); and
    \item \textbf{Case 2.} the secret $S$ is arbitrarily distributed but the leakage channels are independent binary symmetric channels (BSCs) (Proposition~\ref{prop: security_lemma}).
\end{itemize}
The relationship between mutual information security (MIS), and other well-known notions of security in the literature i.e. semantic security(SS), and distinguishing security (DS) (Theorem~\ref{thm: inequality_equivalence}) is used in Corollary~\ref{cor: security_metrics}. 

\subsection{Shamir's Scheme with General Threshold}

We begin by showing that the analysis of the security guarantees for the Shamir's secret sharing scheme can be reduced to the scenario with no colluding users.

\begin{lem}\label{lemma: no_collusion}
$\text{ShamirSS}(N,t)$ scheme with $t'<t$ colluding users, under the statistical leakage model introduced in Section\,\ref{sec: Leakage Model}, is  as secure as the $\text{ShamirSS}(N-t',t-t')$ scheme with no colluding users. 
\end{lem}
\begin{proof}
    The problem of recovering the secret in $\text{ShamirSS}(N,t)$ is equivalent to the problem of recovering the corresponding $(t-1)$-degree polynomial from its evaluation on $N$ distinct points in the space $\mathbb{F}_q[x]$. The problem is equivalent to the Chinese Remainder Theorem for rings and reduces to finding the unique element in the ring $\mathbb{F}_q[x]/(x^t)$ which satisfies the given $N$ evaluation points. Using the shares of the colluding users, we can interpolate and find the unique monic polynomial $h(x)$ of degree $t'$ such that it satisfies the interpolation points of the colluding users. The problem is therefore reduced from finding unique element in the ring $\mathbb{F}_q[x]/x^t$, to finding a unique element in the ring $\left(\mathbb{F}_q[x]/(x^t) \right)/(h(x))$ satisfying the $N-t'$ evaluation points. This is equivalent to finding a unique element in the ring $\mathbb{F}_q[x]/(x^{t-t'})$ satisfying $N-t'$ evaluation points which in turn is equivalent to the $\text{ShamirSS}(N-t',t-t')$ scheme. 
\end{proof}

\begin{cor}
    The security analysis of Shamir's scheme with $t'$ colluding users can be reduced to a scheme with no collusion.
\end{cor}

\begin{prop}\label{prop: general_security}
     The mutual information leakage in ShamirSS$(N, t)$ with $t'$ colluding users is bounded as
    \begin{equation}
        \frac{I(S; \mathbf{Z})}{l} \le \hspace{1mm} \sum_{i=1}^{N-t'} \epsilon_i,
    \end{equation}
    where $\epsilon_i$'s are the constants characterizing the per bit information leaked from the honest secret shares i.e., $I(\mathcal{U}_i,Z_i)/l \le \epsilon_i$.
\end{prop}

\begin{proof}
Consider a set $\mathcal{U}' \equiv \{ \mathcal{U}_1,\ldots,\mathcal{U}_{t-1} \}$ of $t-1$ secret shares which contains the secret shares of the $t'$ colluding users.  Since $S - (\mathcal{U}_1,\ldots,\mathcal{U}_{t-1}, Z_t,\ldots,Z_N) - (Z_1,\ldots,Z_N)$ is a Markov chain, using the data processing inequality,
    \begin{align*}
        I(S; \mathbf{Z}) &= I\left(S; (Z_1,\ldots,Z_N)\right) \\
        &\le I\left(S; (\mathcal{U}_1,\ldots,\mathcal{U}_{t-1}, Z_{t}, \ldots, Z_N)\right) \\
        &\le I\left(S; \mathcal{U}' \right) + I\left(S; ( Z_{t}, \ldots, Z_N)/\mathcal{U}'\right) \\
        &\le I\left(S; ( Z_{t}, \ldots, Z_N)/\mathcal{U}'\right)
    \end{align*}
    But $I\left(S; ( Z_{t}, \ldots, Z_N)/\mathcal{U}'\right)$ is the information leaked from the secret in ShamirSS$(N,t)$ with $t-1$ colluding users. Using Lemma~\ref{lemma: no_collusion}, this scheme is as secure as the ShamirSS$(N-t+1,1)$ scheme with the same adversary leakage channels. The ShamirSS$(N-t+1,1)$ is a scheme where all the secret shares are the secret $S'$ itself, and the adversary obtains noisy versions $\mathbf{Z}' = (Z_t',\ldots,Z_N')$ of the secret through all the N-t+1 users. Therefore, the information leaked in the ShamirSS$(N,t)$ scheme can be bounded as 
    \begin{align}
    I(S;\mathbf{Z}) &\le
    I\left(S; ( Z_{t}, \ldots, Z_N)/\mathcal{U}'\right) \\ 
    &= I (S'; \mathbf{Z'} = (Z'_t,\ldots,Z'_{N})) \\ \label{eqn: N-t+1 independece 1}&= H(\mathbf{Z'}) - H(\mathbf{Z'}/S') \\
    \label{eqn: N-t+1 independece 2}&=  H(\mathbf{Z'}) - \sum_{i=t}^{i=N}H(Z'_i/S') \\
    &\le \sum_{i=t}^{i=N} \left( H(Z'_i) - H(Z'_i/S)\right) \\
    &\le \sum_{i=t}^{i=N} I(S';Z'_i) \\
    &\le \sum_{i=1}^{i=N-t'} l \cdot \epsilon_i
    \end{align}
    Here \eqref{eqn: N-t+1 independece 1}$\implies$\eqref{eqn: N-t+1 independece 2} due to the independence of the $N-t+1$ leakage channels given all the secret shares are $S'$.
\end{proof}

\begin{cor}\label{cor: loos_security_metrics}
    The MIS-security metric can be characterized as $\eta_{MIS} \le \sum_{i=1}^{i=N-t'} l \cdot \epsilon_i$. Using Theorem~\ref{thm: inequality_equivalence}, semantic and distinguishing security metrics can be bounded as $\eta_{SS} \le \eta_{DS} \le \sqrt{2l \left(\sum_{i=1}^{i=N-t'} \epsilon_i\right)}$.
\end{cor}
\subsection{The Extreme-Threshold Scenario}

In general, it might be possible to improve upon the upper bound on the leakage provided in Proposition\,\ref{prop: general_security}. In order to demonstrate this, in the remainder of this section, we consider an extreme scenario of the Shamir's scheme, where the $\text{ShamirSS}(N,N)$ scheme is deployed, and the adversary has access to a noisy version of all the secret shares. In subsequent analysis, in light of Lemma\,\ref{lemma: no_collusion}, it is sufficient to consider scenarios without colluding users, that is, assume that the adversary only has a noisy observation of each secret share. Furthermore, we will assume that each secret share is leaked with $I(\mathcal{U}_j;\mathbf{Z})/l \le \epsilon$ per bit leakage rate for some $\epsilon < 1$. Recall from Section~\ref{subsec:ShamirSS}, in~\eqref{eqn: ShamirSS(N,t) coefficients}, once someone has access to the evaluation points, the secret recovery equation for ShamirSS$(N, N)$ can be uniquely written as 
\begin{equation}\label{equation: secret_linear_dependence}
    s = \sum_{i=1}^{N}c_is_i,
\end{equation}
for some $c_i \in \mathbb{F}_q$. The unique representation of the secret as a linear combination of all the secret shares is a consequence of the following system of equations having a unique solution:
\begin{equation}
  \begin{bmatrix}
    1 & \gamma_1 & \gamma_{1}^2 & \cdots & \gamma_{1}^{N-1} \\
    1 & \gamma_2 & \gamma_{2}^2 & \cdots & \gamma_{2}^{N-1} \\
    \vdots &\vdots & \vdots & &\vdots \notag \\
    1 & \gamma_N & \gamma_{N}^2 & \cdots & \gamma_{N}^{N-1}
  \end{bmatrix}
  \begin{bmatrix}
      s \\
      p_1 \\
      \vdots \notag \\
      p_{N-1}
  \end{bmatrix} = 
  \begin{bmatrix}
      s_1 \\
      s_2 \\
      \vdots \notag \\
      s_{N}
  \end{bmatrix}.
\end{equation}

\begin{rem}\label{rem: choosing_V}
    The coefficients $c_i$'s in \eqref{equation: secret_linear_dependence}, i.e. the secret recovery equation for ShamirSS$(N,N)$, are the elements of the first row in the inverse of the Vandermonde matrix: 
    \begin{equation*}
        V^{-1} = 
        \begin{bmatrix}
    1 & \gamma_1 & \gamma_{1}^2 & \cdots & \gamma_{1}^{N-1} \\
    1 & \gamma_2 & \gamma_{2}^2 & \cdots & \gamma_{2}^{N-1} \\
    \vdots &\vdots & \vdots & &\vdots \notag \\
    1 & \gamma_N & \gamma_{N}^2 & \cdots & \gamma_{N}^{N-1}
  \end{bmatrix}^{-1}. 
    \end{equation*}
    Therefore, replacing the matrix $V$ by any invertible matrix gives us a linear secret sharing scheme with a unique secret recovery equation whose coefficients are determined by the first row of the inverse of the chosen matrix. If the chosen matrix is such that first row of its inverse is filled entirely with $1$s, the resulting linear secret sharing scheme has the recovery equation
    \begin{equation}\label{eqn: linear_ss}
        s = s_1 + s_2 + \cdots + s_N.
    \end{equation}
\end{rem}

To bound the information leaked from the secret $S$ by the leaked secret shares, for the two cases, 
\begin{itemize}
    \item \textbf{Case 1.} we will use Mrs. Gerber's Lemma (MGL) \cite{wyner1973theorem} (Lemma~\ref{lemma: Gerbers}) in proving Proposition~\ref{prop: security_MGL}; and for
    \item \textbf{Case 2.} we will use Corollary~\ref{corollary: Markov_chain_MI} in proving Proposition~\ref{prop: security_lemma}. 
\end{itemize}
Note that both the MGL, and Corollary~\ref{corollary: Markov_chain_MI} only hold true for the field $\mathbb{F}_2$. Since addition in the field $\mathbb{F}_{2^l}$ can be done by bit-wise addition of the vector representations over $\mathbb{F}_2$, we only provide bit-wise security guarantees. To bound information leakage for each bit, separate equations for each bit of the secret need to be formed from the secret recovery equation~\eqref{equation: secret_linear_dependence}. But multiplication in the field $\mathbb{F}_{2^l}$ depends on the choice of the primitive (splitting) polynomial. Assuming that the adversary has access to all the coefficients of the secret shares and the splitting polynomial, the adversary has access to the $l$ bit-wise equations for each bit of the secret. Our bounding results, to be presented in Propositions~\ref{prop: security_MGL}, and ~\ref{prop: security_lemma}, depend upon the number of summands in the estimation of the sum of the random variables. Alternatively, one can replace the ShamirSS$(N,N)$ with another linear scheme where all the secret shares, $s_i$'s, appear in the reconstruction of the secret $s$ with coefficient $1$, as discussed in Remark~\ref{rem: choosing_V}. This automatically improves the bounds, which we will see in Remarks~\ref{rem: MGL security} and~\ref{rem: (N,N) security}. 
\begin{rem}\label{rem: tildeN_assumption}
In general, it is not straightforward to characterize the exact number of summands in the $l$ bit-wise equations for the $l$ bits of the secret. Let $\Tilde{N}$ denotes the minimum of the number of non-zero summands across the $l$ bit-wise equations (derived from \eqref{equation: secret_linear_dependence}). We will utilize $\Tilde{N}$ in the expression of the bounds in order to simplify them. Note that given the randomized coefficient selection inherent in Shamir's scheme (Section~\ref{subsec:ShamirSS}), it is expected that $\Tilde{N}$ grows at least linearly with the number of users, i.e., $\Tilde{N} = \Omega(N)$.    
\end{rem}

% \begin{rem}
  
% \end{rem}
%\begin{itemize}
 %   \item \textbf{Assumption. }In the ShamirSS$(N,N)$ scheme being deployed, the number of non-zero summands in every equation for every bit (derived from \eqref{equation: secret_linear_dependence}), is $N/2$.
%\end{itemize}
% \begin{rem}
%     In the analysis of leakage resilience of Shamir's scheme, the authors in \cite{benhamouda2021local, maji2022improved}, consider the Shamir's scheme with fixed evaluation points; while the authors in \cite{maji2021leakage, maji2022leakage} conduct their investigation for Shamir's scheme with random evaluation points. We believe that our assumption is valid with high probability for  a Shamir's scheme with random evaluation points.
% \end{rem}

% \begin{lemma}[\cite{hirche2020renyi}, Lemma III.7]
%     Let $X_0$, $X_1$ be independent $\mathbb{Z}_2$-valued random variables with side information $\mathbf{Y} = (Y_0,Y_1)$, and $Z = X_0 \oplus X_1$. Then
%         \begin{equation}
%         H_{\infty}^{A}(Z |Y) = h_{\infty} \left( h_{\infty}^{-1}\left(H_{\infty}^{A}(X_0|Y_0) \right) \star h_{\infty}^{-1}\left(H_{\infty}^{A}(X_1|Y_1) \right) \right).
%     \end{equation}
% \end{lemma}
\subsection{Security Analysis}
\begin{lem}[Mrs. Gerber's Lemma, \cite{wyner1973theorem}]\label{lemma: Gerbers}
    Let $A_0$, $A_1$ be independent $\mathbb{Z}_2$-valued random variables with side information $\mathbf{B} = (B_0,B_1)$ (i.e. $B_0-A_0-A_1-B_1$ is a Markov chain) , and $C = A_0 \oplus A_1$. Then
    \begin{align}
        H(C|B) \ge h \left( h^{-1}\left(H(A_0|B_0) \right) \star h^{-1}\left(H(A_1|B_1) \right) \right),
    \end{align}
    where $a \star b = a(1-b) + b(1-a)$.
\end{lem}

%\begin{rem}
%\textcolor{red}{[\cite{}]}
 %   The convexity result used to prove the Mrs. Gerber's lemma is true for all abelian groups of order $2^n$, and conjectured to be true for all finite abelian groups.
%\end{rem}

\begin{prop}\label{prop: security_MGL}
    Given that the secret $S$ is a uniform random variable in $F_{2^l}$, the mutual information leakage from each bit of the secret in ShamirSS$(N, N)$ with $t'$ colluding users is bounded as,
    \begin{equation}
        I(S^i; \mathbf{Z}) \le \hspace{1mm} \delta^{2(\Tilde{N}-t')},
    \end{equation}
    where $S^i$ is the random variable characterizing the $i^{th}$ bit of the secret, $\Tilde{N}$ is defined in Remark~\ref{rem: tildeN_assumption}, and $\delta = 1 - 2 h^{-1}\left( 1- \epsilon \right)$. 
\end{prop}
\begin{proof}
Using Lemma~\ref{lemma: no_collusion}, we can work with a scheme without colluding users i.e. $\Tilde{N} \to (\Tilde{N}-t')$. The secret $S$ is a linear combination of the secret shares $\mathcal{U}_i$'s in the field $F_{2^l}$. By the assumption on the number of non-zero summands in the $l$ bit equations, let $S^i = \mathcal{u}_1 + \cdots + \mathcal{u}_{\Tilde{N}_i}$ where $\Tilde{N}_i>\Tilde{N}$ is the number of non-zero summands in the equation for $S^i$, and $\mathcal{u}_j$, $j \in \{1, \cdots, \Tilde{N}_i \}$, is a bit in the binary vector representation of some secret share $\mathcal{U}_j$, $j \in \{1, \cdots, N \}$. The secret $S$ is uniformly random, and therefore all the summands $\mathcal{U}_j$'s are independent. Let $h^{-1}\left(H(\mathcal{U}_i|Y_i)/l \right) = q$. Using Mrs. Gerber's Lemma for each bit,
    \begin{align*}
        H(S^i/\mathbf{Z}) &\ge h \left( \star_{1}^{\Tilde{N}} h^{-1}\left(H(\mathcal{U}_i|Z_i)/l \right) \right) \\ 
        &\ge h(q \star q \star \cdots \star q) \\
         &\ge h \left( \frac{1-\delta^{\Tilde{N}}}{2} \right) \\
         &\ge 1- \delta^{2\Tilde{N}}. \\
         I(S^i; \mathbf{Z}) \hspace{1mm} &= H(S^i) - H(S^i/\mathbf{Z}) \\
         &\le \delta^{2\Tilde{N}},
    \end{align*}
where $\delta = 1-2q = 1 - 2 h^{-1}\left( 1- \epsilon \right)$.
\end{proof}
\begin{cor}
    For the ShamirSS$(N,N)$ scheme with $\Tilde{N}>\kappa N$, where $\kappa<1$; the mutual information leaked from each bit of the secret reduces exponentially with the number of users.
\end{cor}
\begin{rem}\label{rem: MGL security}
    If a linear secret sharing scheme is of the form described in Remark~\ref{rem: choosing_V}, with its secret recovery equation as \eqref{eqn: linear_ss}; the bit $S^i$ is the summation of $i^{th}$ bit of each secret share, and the leakage can be bounded as $I(S^i, \mathbf{Z}) \le \delta^{2(N-t')}$. 
\end{rem}
To use Mrs. Gerber's lemma, we need independence of the random variables being added. In this application of Mrs. Gerber's lemma, the summand random variables are the secret shares, and they are independent if and only if the secret $S$ is distributed uniformly. But as seen in Definition~\ref{def: mutual_information_security}, MIS-security requires bounding mutual information leakage for all distributions of the secret. Therefore, even for the Shamir's scheme with a general threshold, ShamirSS$(N,t)$, Mrs. Gerber's Lemma cannot be used to provide MIS-security guarantee. To understand the leakage resilience of ShamirSS$(N,N)$, we consider the case of leakage through BSCs. 
    
\begin{lem}[]\label{lemma: (N,N)}
    Let $\mathbf{X} = (X_0, \cdots, X_{k-1})$ be $\mathbb{Z}_2$-valued uniformly distributed random variables with every $k-1$ of them being mutually independent. Let $\mathbf{Y} = (Y_0,\cdots,Y_{k-1})$ be the output of a BSC with $\mathbf{X}$ as the input. Then 
    \begin{equation}\label{eqn: lemma (N,N)}
        (X_0+\cdots+X_{k-1}) - (Y_0 + \cdots + Y_{k-1}) - \mathbf{Y}  
    \end{equation}
    is a Markov chain.
    % \begin{equation}
    %     I\left((X_0+\cdots+X_{N-1}); \mathbf{Y} \right) = I\left((X_0+\cdots+X_{N-1}); (Y_0 + \cdots + Y_{N-1})\right)
    % \end{equation}
\end{lem}

\begin{proof}
Define the function $S(X) = X_0 + \cdots + X_{k-1}$, where $X$ is some realization of the random variable $\mathbf{X}$. Since $X$ is a binary vector, we can also look at the Hamming weight of the vector denoted by $wt(X)$. Note that while $S(X)$ is a binary variable i.e. $S(X) \in \mathbb{Z}_2$, while $wt(X)$ is a positive integer i.e. $wt(X) \in \mathbb{Z}$. Let the independent BSCs have parameter $q$ i.e. the probability that each bit flips for each bit $X_i$ is $q$. Since each subset of size $(k-1)$ of the bits of the random vector $\mathbf{X}$ is independent, we must have that 
    \begin{align*}
        \frac{1}{2^{k-1}} = \   &Pr(X_0=0, X_1 = 0, \cdots, X_{k-1}=0) \\&+ Pr(X_0=1, X_1 = 0, \cdots, X_{k-1}=0) \\ \frac{1}{2^{k-1}} 
        = \ &Pr(X_0=0, X_1 = 0, \cdots, X_{k-1}=0)\\ &+ Pr(X_0=0, X_1 = 1, \cdots, X_{k-1}=0) \\
       \frac{1}{2^{k-1}} = \ &\hspace{3mm} \cdots
    \end{align*} \\
    Therefore,
    \begin{align*}
    Pr(\mathbf{X} = (i_0, i_1, \cdots ,i_{k-1}))&= Pr(\mathbf{X} = (j_0, j_1, \cdots ,j_{k-1})) \\ \iff i_0 + \cdots i_{k-1} &\equiv j_0 + \cdots j_{k-1} \hspace{-3mm} \mod{2}.  
    \end{align*}
    Define $Pr(\mathbf{X} = (i_0, i_1, \cdots ,i_{k-1})) = p_{S(X)}$ and observe that
    \begin{align*}
        Pr(\mathbf{Y} = (y_0, \ldots y_{k-1}), S(X) = 0) = \\\sum_{X \text{ with } S(X)=0} p_0 q^{wt(X-Y)}(1-q)^{k-wt(X-Y)}.
    \end{align*}
    %= p_{i_0 + \cdots i_{k-1} \hspace{-2mm}\mod{2}} =
    This is true because exactly $wt(X-Y)$ bits need to be flipped for the input to be $X$, and for the BSC(q) to output $Y = (y_0, \ldots y_{k-1})$ where $Y$ is some realization of the random variable $\mathbf{Y}$. Furthermore, note that for all values of $Y$ such that $S(Y) = 0$, the right hand side remains constant. Therefore,
    \begin{align*}
        &Pr(\mathbf{Y} = (0, \ldots 0) / S(X) = 0, S(Y)=0) \\= &\frac{\sum_{X \text{ with } S(X)=0} p_0 q^{wt(X-Y)}(1-q)^{k-wt(X-Y)}}{\sum_{Y \text{ with } S(Y)=0}\sum_{X \text{ with } S(X)=0} p_0 q^{wt(X-Y)}(1-q)^{k-wt(X-Y)}} \\
        = &\frac{1}{2^{k-1}}
        \\ = &Pr(\mathbf{Y} = (0, \ldots 0) / S(X) = 1, S(Y)=0).
    \end{align*}
    But this implies that $(X_0+\cdots+X_{k-1}) - (Y_0 + \cdots + Y_{k-1}) - \mathbf{Y}$ is a Markov chain.

\end{proof}

\begin{cor}\label{corollary: Markov_chain_MI} \begin{align*}
         &I\left((X_0+\cdots+X_{k-1}); \mathbf{Y} \right) = \\  &I\left((X_0+\cdots+X_{k-1}); (Y_0 + \cdots + Y_{k-1})\right).
    \end{align*}
\end{cor}
\begin{proof}
 From Lemma~\ref{lemma: (N,N)}, we know that $(X_0+\cdots+X_{k-1}) - (Y_0 + \cdots + Y_{k-1}) - \mathbf{Y}$ is a Markov Chain. Since $(Y_0 + \cdots + Y_{k-1})$ is a function of $\mathbf{Y}$, $\mathbf{}$ $(X_0 + \ldots + X_{k-1}) - \mathbf{Y} - (Y_0 + \cdots + Y_{k-1})$ is also a Markov chain. The corollary now follows from using the data processing inequality twice.
\end{proof}

%\begin{conj}\label{conj: lemma_generalization}
    %Lemma~\ref{lemma: (N,N)} holds true for all abelian groups.
%\end{conj}

% \begin{rem}
%     Since the result discussed in corollary~\ref{corollary: Markov_chain_MI} only pertains to addition in the binary field, the result cannot be directly applied for the security analysis of the ShamirSS$(N, N)$ scheme in the field $\mathbb{F}_{2^l}$. 
% \end{rem}

\begin{prop}\label{prop: security_lemma}
     Given that all the secret shares leak information independently for every bit through BSC(q); the mutual information leakage from each bit of the secret in ShamirSS$(N, N)$ with $t'$ colluding users is bounded as 
    \begin{equation}
        I(S^i; \mathbf{Z}) \le \hspace{1mm} \delta^{2(\Tilde{N}-t')},
    \end{equation}
    where $S^i$ is the random variable characterizing the $i^{th}$ bit of the secret, $\Tilde{N}$ is defined in Remark~\ref{rem: tildeN_assumption}, and $\delta = 1 - 2 h^{-1}\left( 1- \epsilon \right)$.
\end{prop}

\begin{proof}
Using Lemma~\ref{lemma: no_collusion}, we can work with a scheme without colluding users with $\Tilde{N} \to (\Tilde{N}-t')$. By the assumption on the number of non-zero summands in the $l$ bit equations, let $S^i = \mathcal{u}_1 + \cdots + \mathcal{u}_{\Tilde{N}_i}$ where $\Tilde{N}_i>\Tilde{N}$ is the number of non-zero summands in the equation for $S^i$, and $\mathcal{u}_j$, $j \in \{1, \cdots, \Tilde{N}_i \}$ is some digit of some secret share $\mathcal{U}_k$, $k \in \{1, \cdots, N \}$. Let the noise added to the bit $\mathcal{u}_j$ in the relevant secret share be $o_j$, i.e. $z_j = \mathcal{u}_j + o_j$. Since any $(N-1)$ secret shares are uniformly distributed and mutually independent, using Corollary~\ref{corollary: Markov_chain_MI} for each bit,
\begin{align*}
    I(S^i; \mathbf{Z})/l &= I(S; z_1 + \ldots+ z_{\Tilde{N}}) \\
    &= H(z_1+\ldots+z_{\Tilde{N}}) - H(o_1 + \ldots + o_{\Tilde{N}}) \\
    &\le H(h^{-1}(H(S^i)) \star q \star \cdots \star q) - H( q \star \cdots \star q) \\
    & \le 1 - h \left( \frac{1 - (1-2q)^{\Tilde{N}}}{2} \right) \\
    & \le (1-2q)^{2\Tilde{N}} \\
    & \le \delta ^{2\Tilde{N}},
\end{align*}
    where $q = h^{-1}(1-\epsilon)$ and $\delta = 1 - 2 h^{-1}\left( 1- \epsilon \right)$.

\end{proof}
\begin{cor}\label{cor: security_metrics}
    Since Proposition~\ref{prop: security_lemma} is valid for all distributions of the secret share, the bit-wise MIS-security metric can be characterized as $\eta_{MIS} \le \delta^{2(\Tilde{N}-t')}$. Using Theorem~\ref{thm: inequality_equivalence}, semantic and distinguishing security metrics for each bit of the secret can be bounded as $\eta_{SS} \le \eta_{DS} \le \sqrt{2}\delta^{\Tilde{N}-t'}$. For $\Tilde{N}>\kappa N$, where $\kappa<1$, the bit-wise security-metrics improve exponentially with the number of users.
\end{cor}
\begin{rem}\label{rem: (N,N) security}
   If a linear secret sharing scheme is of the form described in Remark~\ref{rem: choosing_V}; the bit $S^i$ is the summation of $i^{th}$ bit of each secret share, and the security metrics can be bounded as $\eta_{MIS} \le \delta^{2(N-t')}$, and $\eta_{SS} \le \eta_{DS} \le \sqrt{2}\delta^{N-t'}$. 
\end{rem}

\section{Conclusion}
\label{sec:conc}
To examine the leakage resilience of secret sharing schemes in distributed systems, we proposed the \textit{statistical leakage} model. This is an information-theoretic model of leakage, characterized by the honest parties leaking information to an adversary through independent wiretap channels. We then study the leakage resilience of Shamir's secret sharing with statistical leakage. For ShamirSS$(N,N)$ scheme over fields of characteristic $2$, we show that the bit-wise mutual information security (MIS), and consequently,  the semantic security (SS) and distinguishing security (DS), improve exponentially with the number of users. The leakage model proposed in this work can be adapted to understand the leakage resilience of other protocols, such as for the recently proposed secret sharing in analog domain \cite{soleymani2022analog}, and for multi-user secret sharing \cite{soleymani2020distributed}.

There are several directions for future research. We believe that Lemma~\ref{lemma: (N,N)} can be generalized to all Abelian groups, which lets us generalize Proposition~\ref{prop: security_lemma} beyond bit-wise security. Proposition~\ref{prop: security_lemma} provides security guarantees when the leakage channels are BSCs; similar results should hold true for other channels. For ShamirSS$(N,t)$, the secret has $\binom{N}{t}$ secret recovery equations. By generalizing Lemma~\ref{lemma: (N,N)} to every subset of $t$ random variables being mutually independent for a secret $S$ satisfying a system of linear equations, it might be possible to improve upon the upper bound on the leakage provided in Proposition\,\ref{prop: general_security}, to grow exponentially with the threshold $t$.

\bibliographystyle{IEEEtran}
\bibliography{bibliography}
\end{document}